\documentclass[journal]{IEEEtran}

\usepackage{array}
\usepackage{multirow}
\usepackage{amsmath,amssymb,amsthm}
\usepackage{xcolor}
\usepackage{graphicx}
\usepackage{cite}
\usepackage{algorithm}
\usepackage{algpseudocode}
\usepackage{hyperref}
\usepackage{stfloats}
\usepackage{microtype}
\usepackage{tikz}
\usetikzlibrary{patterns.meta,positioning,arrows.meta,calc,fit}
\usepackage{pgfplots}
\pgfplotsset{compat=1.18}
\hypersetup{hidelinks}
\usepackage[caption=false,font=footnotesize]{subfig}

\newtheorem{remark}{Remark}
\newtheorem{proposition}{Proposition}
\newtheorem{approximation}{Approximation}
\newtheorem{problem}{Problem}

\newcommand{\ct}{\mathbf{c}^{\top}}
\newcommand{\Prb}{\mathbb{P}}
\newcommand{\E}{\mathbb{E}}

\makeatletter
\g@addto@macro\normalsize{%
  \setlength{\abovedisplayskip}{3pt plus 1pt minus 1pt}%
  \setlength{\belowdisplayskip}{3pt plus 1pt minus 1pt}%
  \setlength{\abovedisplayshortskip}{1pt plus 1pt}%
  \setlength{\belowdisplayshortskip}{2pt plus 1pt minus 1pt}%
}
\makeatother

\definecolor{tealcolor}{RGB}{13,140,153}
\definecolor{ambercolor}{RGB}{217,137,20}
\definecolor{greencolor}{RGB}{8,135,93}
\definecolor{navycolor}{RGB}{27,58,92}
\definecolor{slatecolor}{RGB}{61,81,102}
\definecolor{curvecol}{HTML}{1B3A5C}
\definecolor{xcol}{HTML}{2A9D8F}
\definecolor{disruptcol}{HTML}{E9A820}
\definecolor{recovcol}{HTML}{5FBCD3}

\begin{document}
\bstctlcite{IEEEexample:BSTcontrol}

\title{Predictive Triggering for Outage-Resilient Threshold Decisions over Short-Packet Links}

\author{Nho-Duc~Tran,~\IEEEmembership{Student Member,~IEEE,}
        Aamir~Mahmood,~\IEEEmembership{Senior Member,~IEEE,}
        and~Mikael~Gidlund,~\IEEEmembership{Fellow,~IEEE}
\thanks{T.N. Duc, A. Mahmood, and M. Gidlund are with the Department of Computer and Electrical Engineering,  Mid Sweden University, Homlgatan 10, 851 70 Sundsvall, Sweden (e-mail: \{nhoduc.tran, aamir.mahmood, mikael.gidlund\}@miun.se).}
}

\maketitle

\begin{abstract}
Remote threshold decisions require more than accurate state estimates: the posterior must support reliable alarm/no-alarm decisions and, when possible, anticipate early critical decisions. We study this problem over short-packet wireless links with outage risk. We derive false-positive/false-negative feasibility conditions that define a decision-feasible region of the estimation and yield a predictive decision-update trigger. To protect predictive updates from outages, we add AoI-controlled resilience updates that both detect disruptions and maintain freshness. A two-state Markov surrogate of the thresholded process, matched to its one-step switching statistics, enables tractable long-term reliability-energy analysis. Then, we jointly optimized transmit power and AoI-controlled resilience update probabilities. Simulations show earlier, reliable decisions at competitive energy with baselines.
\end{abstract}


\IEEEpeerreviewmaketitle

\section{Introduction}

Many remote monitoring systems (e.g., safety monitoring, anomaly detection) make critical binary decisions by comparing a state-dependent quantity with a threshold. In such systems, the remote decision agent does not observe the state directly, but relies on sporadic sensor updates over a wireless link, as shown in Fig.~\ref{fig:main_idea}. This makes communication decision-oriented rather than purely estimation-oriented: a low-MSE estimate may still be unsafe for thresholding if the posterior mass lies around the threshold, while a less accurate estimate far from the threshold can already support a reliable alarm/no-alarm action (c.f., Fig.~\ref{fig:example-region}). For safety-critical monitoring, the useful update is often the one delivered before the threshold crossing; otherwise, it may have limited value~\cite{data-lost}. This motivates predictive triggering $\mathcal{P}$, where the sensor transmits when its local posterior certifies an imminent and reliable decision change, rather than waiting only for a large innovation or an already observed transition.

However, predictive triggering is fragile when communication is unreliable. Short-packet wireless links have non-negligible packet error rates (PER), and practical deployments may also experience longer outages due to energy depletion, hardware faults, deep shadowing, or interference. A correctly predicted transition update may then be blocked until it is no longer useful. 
Thus, sending resilient updates $\mathcal{R}$ is essential to maintain freshness, enabling the system to detect and recover from outages before the next predictive transmission. Related work has co-designed sampling and decisions through Age of Information (AoI)~\cite{aoi-aware-mdp}, threshold protocols for distributed classification~\cite{multi-label}, remote estimation over collision channels~\cite{threshold-transmit}, and ideal-link transmissions with innovation magnitude~\cite{confident-level}. These studies improve estimation accuracy, freshness, or event reporting, but do not directly couple explicit false-positive/false-negative rate requirements with early delivery over short-packet links under persistent random outages. This gap motivates a transmission policy that is both decision-aware, by checking whether the posterior is feasible for a reliable binary action, and resilience-aware, by maintaining enough freshness to detect and recover from outages.


In this letter, we focus on threshold decisions rather than full-state reconstruction. Regardless of how the remote agent processes received packets, an update is useful only if the posterior can safely support the required binary action and if this action reaches the agent before the state crosses the threshold. Our contributions are: we derive posterior-based false-positive and false-negative feasibility conditions that define when a threshold decision is reliable; we model persistent random outages and introduce AoI-controlled resilience updates for timely disruption detection and freshness maintenance; we construct a two-state Markov surrogate of the thresholded Gaussian process for tractable long-term reliability and energy analysis; and we jointly optimize the transmit power and state-dependent AoI thresholds under lead-time reliability constraints, validating the resulting policy against standard baselines.

\definecolor{alarmred}{RGB}{185,55,45}
\definecolor{safeblue}{RGB}{45,95,170}
\definecolor{resgreen}{RGB}{45,135,80}
\definecolor{aoiorange}{RGB}{220,130,30}
\definecolor{alarmcol}{HTML}{D62828}
\definecolor{confusecol}{HTML}{F77F00}
\definecolor{safecol}{HTML}{2A9D8F}
\definecolor{deltacol}{HTML}{333333}

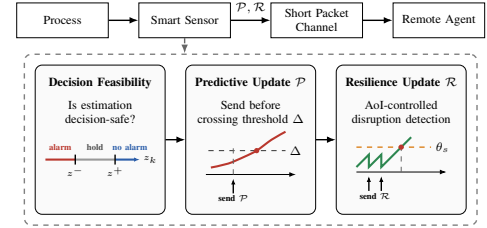
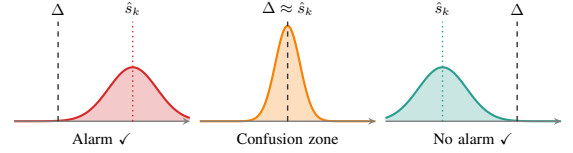
\begin{figure}[tb]
\centering

\subfloat[The smart sensor checks decision feasibility, sending predictive packets $\mathcal P$ before a certified threshold crossing, and frequent resilience packets $\mathcal R$ to maintain freshness for outage detection.]{%
    \makebox[\columnwidth][c]{
        \resizebox{0.7\columnwidth}{!}{%
        \begin{tikzpicture}[
            font=\scriptsize,
            >=Latex,
            sys/.style={
                draw,
                align=center,
                minimum height=7mm,
                text width=1.7cm,
                inner sep=2pt,
                fill=white
            },
            panel/.style={
                draw,
                rounded corners=4pt,
                align=center,
                minimum height=30mm,  
                text width=2.4cm,       
                inner sep=3pt,
                fill=black!2
            },
            arr/.style={-{Latex[length=1.8mm]}, thick},
            softarr/.style={-{Latex[length=1.6mm]}, thin},
            smallarr/.style={-{Latex[length=1.2mm]}}
        ]
        
        \node[panel] (pred) at (0,0) {};
        \node[panel, left=0.4cm of pred] (feas) {};
        \node[panel, right=0.4cm of pred] (res)  {};
        
        \draw[arr] (feas) -- (pred);
        \draw[arr] (pred) -- (res);
        
        \node[draw, dashed, thick, black!50, rounded corners=4pt, inner sep=6pt, fit=(feas) (pred) (res)] (logicbox) {};
        
        \coordinate (topcenter) at ([yshift=2.4cm]pred.center);
        \node[sys, left=0.4cm of topcenter] (sensor) {Smart Sensor};
        \node[sys, left=0.6cm of sensor] (proc) {Process};
        \node[sys, right=0.4cm of topcenter] (link) {Short Packet\\Channel};
        \node[sys, right=0.6cm of link] (agent) {Remote Agent};
        
        \draw[arr] (proc) -- node[above] { } (sensor);
        \draw[arr] (sensor) -- node[above] {$\mathcal P,\mathcal R$} (link);
        \draw[arr] (link) -- node[above] { } (agent);
        
        \draw[arr, dashed, thick, black!60] (sensor.south) -- (sensor.south |- logicbox.north);
        
        \node[font=\scriptsize\bfseries, align=center, anchor=north] at ([yshift=-4pt]feas.north)
        {Decision Feasibility};
        
        \node[align=center, anchor=north] at ([yshift=-17pt]feas.north)
        {Is estimation\\decision-safe?};
        
        \begin{scope}[shift={([xshift=-1.1cm, yshift=-0.4cm]feas.center)}]
            \draw[alarmred, very thick] (0,0) -- (0.6,0);
            \draw[black!40, very thick] (0.6,0) -- (1.4,0);
            
            \draw[smallarr, safeblue, very thick] (1.4,0) -- (1.9,0) node[right=-2pt, text=black, font=\tiny] {$z_k$};
        
            \draw[thick] (0.6,-0.1) -- (0.6,0.1);
            \draw[thick] (1.4,-0.1) -- (1.4,0.1);
        
            \node[below, font=\tiny] at (0.6,-0.02) {$z^-$};
            \node[below, font=\tiny] at (1.4,-0.02) {$z^+$};
        
            \node[alarmred, font=\tiny\bfseries] at (0.3,0.25) {alarm};
            \node[black!70, font=\tiny\bfseries] at (1.0,0.25) {hold};
            \node[safeblue, font=\tiny\bfseries] at (1.7,0.25) {no alarm};
        \end{scope}

        \node[font=\scriptsize\bfseries, align=center, anchor=north] at ([yshift=-4pt]pred.north)
        {Predictive Update $\mathcal P$};
        
        \node[align=center, anchor=north] at ([yshift=-17pt]pred.north)
        {Send before\\ crossing threshold $\Delta$};
        
        \begin{scope}[shift={([xshift=-0.9cm, yshift=-0.70cm]pred.center)}]
            \draw[smallarr] (0,0) -- (1.8,0); 
            \draw[dashed, black!80] (0.05,0.48) -- (1.6,0.48) node[right=-1pt, text=black] {$\Delta$};
        
            \draw[alarmred, very thick] plot[smooth] coordinates {(0.10,0.14) (0.45,0.22) (0.78,0.36) (1.02,0.48) (1.32,0.70) (1.60,0.86)};
        
            \draw[dashed, black!60] (0.55,0) -- (0.55,0.55);
        
            \fill[alarmred] (1.02,0.48) circle (1.5pt);
        
            \draw[smallarr, thick, black] (0.55,-0.38) -- (0.55,-0.08);
            \node[black, font=\tiny\bfseries] at (0.55,-0.48) {send $\mathcal{P}$};
        \end{scope}
        
        \node[font=\scriptsize\bfseries, align=center, anchor=north] at ([yshift=-4pt]res.north)
        {Resilience Update $\mathcal R$};
        
        \node[align=center, anchor=north] at ([yshift=-17pt]res.north)
        {AoI-controlled\\disruption detection};
        
        \begin{scope}[shift={([xshift=-0.9cm, yshift=-0.65cm]res.center)}]
            \draw[smallarr] (0,0) -- (1.8,0);
            \draw[dashed, aoiorange, thick] (0,0.5) -- (1.5,0.5) node[right=-1pt, text=black] {$\theta_s$};
        
            \draw[resgreen, very thick] (0,0.1) -- (0.25,0.35) -- (0.25,0.1) -- (0.5,0.35) -- (0.5,0.1) -- (0.9,0.5) -- (1.1,0.7);
        
            \draw[smallarr, thick, black] (0.25,-0.38) -- (0.25,-0.08);
            \draw[smallarr, thick, black] (0.5,-0.38) -- (0.5,-0.08);
            \node[black, font=\tiny\bfseries] at (0.375,-0.48) {send $\mathcal{R}$};
        
            \draw[dashed, black!60] (0.9,0) -- (0.9,0.5);
            \fill[alarmred] (0.9,0.5) circle (1.5pt);
        \end{scope}
        
        \end{tikzpicture}%
        }
    } 
    \label{fig:main_idea}
}\\[-0.1pt]
\subfloat[The estimator declares alarm when $z_k\le z^-$, no alarm when $z_k\ge z^+$, and holds its previous decision in $z^-<z_k<z^+$.]{%
    \makebox[\columnwidth][c]{
        \definecolor{alarmcol}{HTML}{D62828}
        \definecolor{confusecol}{HTML}{F77F00}
        \definecolor{safecol}{HTML}{2A9D8F}
        \definecolor{deltacol}{HTML}{333333}
        \resizebox{0.82\columnwidth}{!}{%
        \begin{tikzpicture}
        \pgfmathsetmacro{\pW}{3.5}
        \pgfmathsetmacro{\pH}{2.4}
        \pgfmathsetmacro{\gap}{0.2}
        \begin{scope}[xshift=0cm]
          \begin{axis}[at={(0cm,0cm)}, anchor=south west,width=\pW cm, height=\pH cm,xmin=0, xmax=4, ymin=0, ymax=1.45,axis lines=none, ticks=none, clip=false,scale only axis]
            \addplot[domain=0:4, samples=100, fill=alarmcol, fill opacity=0.25, draw=none] {0.65 * exp(-((x-2.7)^2)/(2*0.55^2))} \closedcycle;
            \addplot[domain=0:4, samples=100, very thick, color=alarmcol] {0.65 * exp(-((x-2.7)^2)/(2*0.55^2))};
            \draw[thick, gray, ->, >=stealth] (axis cs:0,0) -- (axis cs:4,0);
            \draw[deltacol, thick, dashed] (axis cs:1.0,0) -- (axis cs:1.0,1.2) node[above, font=\small\bfseries, text=black] {$\Delta$};
            \draw[alarmcol, thick, dotted] (axis cs:2.7,0) -- (axis cs:2.7,1.2) node[above, font=\small\bfseries, text=black] {$\hat{s}_k$};
          \end{axis}
          \node[text=black, font=\small, align=center] at (\pW/2, -0.35) {Alarm \checkmark};
        \end{scope}
        \begin{scope}[xshift=\pW cm + \gap cm]
          \begin{axis}[at={(0cm,0cm)}, anchor=south west,width=\pW cm, height=\pH cm,xmin=0, xmax=4, ymin=0, ymax=1.45,axis lines=none, ticks=none, clip=false,scale only axis]
            \addplot[domain=0:4, samples=100, fill=confusecol, fill opacity=0.25, draw=none] {1.15 * exp(-((x-2.0)^2)/(2*0.28^2))} \closedcycle;
            \addplot[domain=0:4, samples=100, very thick, color=confusecol] {1.15 * exp(-((x-2.0)^2)/(2*0.28^2))};
            \draw[thick, gray, ->, >=stealth] (axis cs:0,0) -- (axis cs:4,0);
            \draw[deltacol, thick, dashed] (axis cs:2.0,0) -- (axis cs:2.0,1.2) node[above, font=\small\bfseries, text=black] {$\Delta \approx \hat{s}_k$};
          \end{axis}
          \node[text=black, font=\small, align=center] at (\pW/2, -0.35) {Confusion zone};
        \end{scope}
        \begin{scope}[xshift=2*\pW cm + 2*\gap cm]
          \begin{axis}[at={(0cm,0cm)}, anchor=south west,width=\pW cm, height=\pH cm,xmin=0, xmax=4, ymin=0, ymax=1.45,axis lines=none, ticks=none, clip=false,scale only axis]
            \addplot[domain=0:4, samples=100, fill=safecol, fill opacity=0.25, draw=none] {0.65 * exp(-((x-1.3)^2)/(2*0.55^2))} \closedcycle;
            \addplot[domain=0:4, samples=100, very thick, color=safecol] {0.65 * exp(-((x-1.3)^2)/(2*0.55^2))};
            \draw[thick, gray, ->, >=stealth] (axis cs:0,0) -- (axis cs:4,0);
            \draw[safecol, thick, dotted] (axis cs:1.3,0) -- (axis cs:1.3,1.2) node[above, font=\small\bfseries, text=black] {$\hat{s}_k$};
            \draw[deltacol, thick, dashed] (axis cs:3.0,0) -- (axis cs:3.0,1.2) node[above, font=\small\bfseries, text=black] {$\Delta$};
          \end{axis}
          \node[text=black, font=\small, align=center] at (\pW/2, -0.35) {No alarm \checkmark};
        \end{scope}
        \end{tikzpicture}%
        }
    } 
    \label{fig:example-region}
}
\caption{Proposed transmission strategy and decision feasibility: $z_k$ denoting the $z$-score relative to the decision threshold, $\hat{s}_k$ the estimated scalar state projection, and $\theta_s$ the state-dependent AoI threshold for outage detection.}
\label{fig:first-page}
\end{figure}


\section{System Model and Problem Formulation}
\label{sec:model}
We consider a sensor monitoring a linear dynamical system with state $\mathbf{x}_k \in \mathbb{R}^N$, while a binary variable $\xi_k \in \{0,1\}$ specifies whether the update is sent at each discrete time step $k$. It is transmitted over a short-packet wireless link 
prone to packet errors and, more critically, random outages 
rendering the link inoperable until recovery is complete. A remote estimator receives these updates and decides whether the scalar projection $s_k \triangleq \mathbf{c}^\top \mathbf{x}_k$,  for a fixed vector $\mathbf{c}$, exceeds a prescribed threshold $\Delta$, issuing an alarm or no-alarm action accordingly.

\subsection{Source and Smart Sensor}
\label{sec:source}
The process and measurements follow
\begin{align}
  \mathbf{x}_{k+1} &= \mathbf{A}\mathbf{x}_k + \mathbf{w}_k,\quad 
    \mathbf{w}_k \sim \mathcal{N}(\boldsymbol{\mu}_w,\mathbf{Q}), \label{eq:state}\\
  \mathbf{y}_k &= \mathbf{C}\mathbf{x}_k + \mathbf{v}_k,\quad 
    \mathbf{v}_k \sim \mathcal{N}(\mathbf{0},\mathbf{R}), \label{eq:measurement}
\end{align}
where $\mathbf{A}\in\mathbb{R}^{N\times N}$ and $\mathbf{C}\in\mathbb{R}^{M\times N}$. 
The noises $\mathbf{w},\, \mathbf{v}$ are uncorrelated; $(\mathbf{A},\mathbf{C})$ is observable, 
$(\mathbf{A},\mathbf{Q}^{1/2})$ is controllable, and $\rho(\mathbf{A})<1$, 
where $\rho(\cdot)$ denotes the spectral radius, ensuring process stability. 
The sensor runs the standard Kalman recursion
\begin{align}
  \hat{\mathbf{x}}^{s}_{k|k-1} &= \mathbf{A}\hat{\mathbf{x}}^{s}_{k-1|k-1},\;
  \mathbf{P}^{s}_{k|k-1} = \mathbf{A}\mathbf{P}^{s}_{k-1|k-1}\mathbf{A}^\top
    +\mathbf{Q},\label{eq:pred}\\
  \mathbf{K}_k &= \mathbf{P}^{s}_{k|k-1}\mathbf{C}^\top
    \bigl(\mathbf{C}\mathbf{P}^{s}_{k|k-1}\mathbf{C}^\top+\mathbf{R}\bigr)^{-1},
    \label{eq:gain}\\
  \hat{\mathbf{x}}^{s}_{k|k} &= \hat{\mathbf{x}}^{s}_{k|k-1}
    +\mathbf{K}_k(\mathbf{y}_k-\mathbf{C}\hat{\mathbf{x}}^{s}_{k|k-1}),
    \label{eq:update}\\
  \mathbf{P}^{s}_{k|k} &= (\mathbf{I}-\mathbf{K}_k\mathbf{C})
    \mathbf{P}^{s}_{k|k-1},\label{eq:cov}
\end{align}
where $\mathbf{K}_k$ is the Kalman gain and $\mathbf{P}^{s}_{k|k}$ is the 
posterior error covariance, which together determine the quality of the local 
state estimate $\hat{\mathbf{x}}^{s}_{k|k}$.

\subsection{Wireless Channel}
For a quasi-static Rayleigh fading channel, the complex coefficient $h_k\sim\mathcal{CN}(0,1)$ remains constant over a packet. The instantaneous received SNR is $\gamma_k=|h_k|^2\bar{\gamma}$,  where $\bar{\gamma}=p_t/\sigma^2$ is the average SNR with transmit power $p_t$  and noise variance $\sigma^2$. With blocklength $n$ and information bits $l$,  the PER is 
$\epsilon_k \approx Q\!\left[\sqrt{n/V(\gamma_k)}\!\left(C(\gamma_k)
-l/n\right)\right],$ with the Shannon capacity $C(\gamma)=\ln(1+\gamma)$ and the channel dispersion $V(\gamma)=1-(1+\gamma)^{-2}$. Letting 
$\varphi=e^{l/n}-1$, $v=e^{-\varphi/\bar{\gamma}}$, and 
$\beta=-\sqrt{n/(2\pi(e^{2l/n}-1))}$, the average PER is 
approximated as~\cite{age-of-loop}
\begin{align}
    \bar{\epsilon} \approx 1 + \!\left( \beta \bar{\gamma} 
    - \beta \bar{\gamma} e^{1/(2\beta \bar{\gamma})} 
    - \tfrac{1}{2} \right) v.
    \label{eq:average_epsilon}
\end{align}
\subsection{Decision Objective}
The binary state is $0$ if $s_k<\Delta$ and $1$ otherwise. The estimator outputs a binary decision $\pi_k\in\{0,1\}$ and must keep decision errors below prescribed levels at each decision instant: the FPR $p_{\mathrm{FP},k} \triangleq  \Prb(\pi_k\!=1\!\mid s_k\!<\!\Delta)\le\alpha_{\mathrm{FP}}$ penalizes spurious alarms,  and the FNR $p_{\mathrm{FN},k}\triangleq\Prb(\pi_k\!=\!0\mid  s_k\ge\Delta)\le\alpha_{\mathrm{FN}}$ penalizes missed alarms.  Beyond accuracy, timeliness is critical. For a true transition at time  $\mathcal{T}$, let $I\triangleq\mathcal{T}-k$ be the sensor-side prediction horizon, i.e., how many steps ahead the sensor anticipates the transition at its local time $k$, and let $L\triangleq\mathcal{T}-\mathcal{T}_u$  be the estimator-side lead time, where $\mathcal{T}_u$ is the time the update is received. Ideally, if the predictive update is delivered without delay,  $L=I\geq0$; however, packet loss or disruptions push $\mathcal{T}_u$ beyond  $\mathcal{T}$, yielding $L<0$, i.e., a late decision. We therefore impose the  lead-time reliability constraint
\begin{align}
    \Prb(L\ge 0)\ge 1-\epsilon_\ell,
    \label{ct:time-trans}
\end{align}
which requires the update to arrive before the transition with high probability 
$1-\epsilon_\ell$. State-indexed versions $I_s, L_s$ refer to sojourns 
in $s\in\{0,1\}$. Here, $L=0$ means that the update arrives exactly at the crossing. This satisfies the formal reliability constraint but leaves no reaction margin. We therefore use $\Prb(L\ge0)$ as the reliability metric and $\Prb(L>0)$ as the actionable-lead-time metric in the numerical results.


\subsection{Remote Agent and Outages}
\label{sec:agent}
Let $\zeta_k\in\{0,1\}$ denote reception success and $\delta_k\triangleq\xi_k\zeta_k$. We consider two remote-agent paradigms because practical threshold-decision systems differ in the amount of information the remote side can process. In \emph{decision adoption}, the sensor transmits $\pi^s_k$ and the agent sets $\pi_k=\pi^s_k$ if $\delta_k=1$, otherwise $\pi_k=\pi_{k-1}$. 
In \emph{filter-based decision}, the sensor sends $(\hat{\mathbf{x}}^s_{k|k},\mathbf{P}^s_{k|k})$, and the agent maintains
\begin{align}
\hat{\mathbf{x}}_k &= \delta_k\hat{\mathbf{x}}^s_{k|k}+(1-\delta_k)\mathbf{A}\hat{\mathbf{x}}_{k-1},\\
\mathbf{P}_k &= \delta_k\mathbf{P}^s_{k|k}+(1-\delta_k)(\mathbf{A}\mathbf{P}_{k-1}\mathbf{A}^\top+\mathbf{Q}).\label{eq:remote_filter}
\end{align}
So it can apply the same decision and prediction rules as the sensor.

Disruptions activate independently at each stage with probability $p_r$, 
uniformly inside the stage. If a disruption starts at $\tau^\star$, reception 
is blocked for its entire duration: $\delta_k=0$ for 
$k=\tau^\star,\ldots,\tau^\star+D_{\theta,s}+t_h-1$, where $D_{\theta,s}$ is 
the detection delay until the AoI at the estimator exceeds the state-dependent 
threshold $\theta_s$ and declares the disruption, and $t_h$ is the subsequent 
recovery time. Here $s\in\{0,1\}$ is the current binary state, allowing 
different detection sensitivities per state. The recovery time follows a 
discrete Weibull distribution, $\Prb(t_h=\tau)=e^{-(\tau/\lambda)^\kappa}-e^{-((\tau+1)/\lambda)^\kappa},$ where $\lambda>0$ controls the mean recovery duration and $\kappa>1$ enforces an increasing hazard rate.
Hence $\theta_s$ simultaneously acts as a disruption detector and a freshness design variable, jointly optimized with transmit power $p_t$.

The overall design goal is to choose $p_t, \theta_0, \theta_1$ to minimize long-run transmission energy while satisfying the lead-time constraint \eqref{ct:time-trans}. Sec.~III-V build the three required ingredients in order: sojourn statistics of the binary process, the posterior condition that triggers a predictive packet, and the outage-resilience model.

\section{Surrogate Two-State Markov Process}
\label{sec:markov}
Evaluating the long-run energy and outage statistics in Sec.~V requires sojourn and transition analysis of the binary process $\{s_k\ge\Delta\}$. Since the decision depends only on whether $s_k$ crosses $\Delta$, we approximate the thresholded process by a two-state Markov chain, reducing the continuous Gaussian dynamics to a tractable structure.
This is justified because, under Sec.~\ref{sec:model}'s assumptions,  $|\rho(\mathbf{A})|<1$ guarantees that $\mathbf{x}_k$ converges to a  stationary distribution $\mathbf{x}_\infty\sim\mathcal{N}(\bar{\mathbf{x}}, \boldsymbol{\Sigma})$, where $\bar{\mathbf{x}}=(\mathbf{I}-\mathbf{A})^{-1} \boldsymbol{\mu}_w$ and $\boldsymbol{\Sigma}$ solves the discrete-time Lyapunov  equation~\cite{anderson2005optimal}. Thus $s_k=\mathbf{c}^\top\mathbf{x}_k$  is a stationary Gaussian with mean $\bar{s}=\mathbf{c}^\top\bar{\mathbf{x}}$  and variance $\sigma_s^2=\mathbf{c}^\top\boldsymbol{\Sigma}\mathbf{c}$, and the switching probabilities are exactly computable from the bivariate normal CDF.
\begin{approximation}\label{thm:switching}
Let binary states $0:=\{s_k<\Delta\}$ and $1:=\{s_k\ge\Delta\}$. The pair $(s_k,s_{k+1})$ is jointly Gaussian with
$\rho_s=\ct\mathbf{A}\boldsymbol{\Sigma}\mathbf{c}/(\ct\boldsymbol{\Sigma}\mathbf{c})$ and $a=(\Delta-\bar s)/\sigma_s$. The one-step switching probabilities are
$q_{01}=1-\Phi_2(a,a;\rho_s)/\Phi(a)$ and
$q_{10}=(\Phi(a)-\Phi_2(a,a;\rho_s))/(1-\Phi(a))$, where $\Phi_2(\cdot,\cdot;\rho)$ is the bivariate standard normal CDF. 
\end{approximation}

These probabilities define the two-state surrogate used for energy and outage analysis in Sec.~V. We next derive when the posterior is concentrated enough to support a reliable binary action - the condition that triggers a predictive update.


\begin{remark}
The thresholded Gaussian process is generally not first-order Markov. The two-state chain is therefore used as a moment-matched surrogate: it preserves the stationary occupation probabilities and the exact one-step switching probabilities of the Gaussian process. Consequently, the mean sojourn lengths are matched through $\mathbb{E}[T_0]=1/q_{01}$ and $\mathbb{E}[T_1]=1/q_{10}$. The Markov approximation enters when the full sojourn tail is replaced by a geometric tail. This is accurate enough for the considered operating regime, where the optimization mainly depends on mean sojourns, short prediction horizons, and average blackout time; Sec.~\ref{subsec:markov_validation} quantifies the residual error.
\end{remark}


\section{Decision Feasibility and Predictive Scheduling}
\label{sec:predictive}
With sojourn statistics in place, we now identify when the posterior is concentrated enough to trigger a predictive $\mathcal{P}$ packet. A threshold decision depends on both proximity to $\Delta$ and posterior uncertainty; the sufficient scalar is $z_k\triangleq(\Delta-\hat{s}_k)/\sigma_k$.

\begin{remark}
State-dependent triggering makes the exact posterior a Gaussian mixture under 
packet drops. For tractability, we approximate the estimator posterior as 
$s_k\!\sim\!\mathcal{N}(\hat{s}_k,\sigma_k^2)$, where $(\hat{s}_k,\sigma_k^2)$ 
summarizes received information up to step $k$.
\end{remark}

\subsection{Decision-Feasibility Region}

\begin{proposition}\label{prop:decision_feasibility}
Under the Gaussian posterior approximation $s_k\sim\mathcal{N}(\hat{s}_k,\sigma_k^2)$, the FPR/FNR is satisfied if $z_k\le z^-\triangleq\Phi^{-1}(\alpha_{\mathrm{FP}})$ for $\pi_k=1$, or $z_k\ge z^+\triangleq\Phi^{-1}(1-\alpha_{\mathrm{FN}})$ for $\pi_k=0$. If $z_k\in(z^-,z^+)$, no binary decision simultaneously satisfies both bounds; this is the confusion region.
\end{proposition}
\begin{proof}
Under $s_k\sim\mathcal{N}(\hat{s}_k,\sigma_k^2)$, we have 
$\Prb(s_k<\Delta)=\Phi(z_k)$. The FPR/FNR 
constraints then reduce to $\Phi(z_k)\le\alpha_{\mathrm{FP}}$ and 
$1-\Phi(z_k)\le\alpha_{\mathrm{FN}}$, equivalent to $z_k\le z^-$ and 
$z_k\ge z^+$. For $\alpha_{\mathrm{FP}},\alpha_{\mathrm{FN}}<0.5$, two regions are disjoint.
\end{proof}
Inside the confusion region, the estimator holds $\pi_k=\pi_{k-1}$, avoiding spurious switches near $\Delta$ at the price of possible delay. At the sensor, however, a deterministic decision is needed to drive the predictive search; we therefore introduce an auxiliary threshold $\phi\in[\Gamma_0,\Gamma_1]$ that resolves the confusion region only, with $\Gamma_1=\Delta-z^-\sigma_p$, $\Gamma_0=\Delta-z^+\sigma_p$, and $\sigma_p^2=\ct\mathbf{P}_\infty\mathbf{c}$.

\begin{proposition}\label{prop:xi_equivalence_exact_fpr_fnr}
The resulting sensor rule is equivalent to $\pi^s_{k}=\mathbf{1}\{\hat s_k^s\ge\phi\}$, with steady-state error rates $\mathrm{FPR}(\phi)=(\Phi(a)-\Phi_2(a,b(\phi);\rho))/\Phi(a)$ and $\mathrm{FNR}(\phi)=(\Phi(b(\phi))-\Phi_2(a,b(\phi);\rho))/(1-\Phi(a))$, where $b(\phi)=(\phi-\bar s)/\sigma_{\hat s}$, $\rho=\sigma_{\hat s}/\sigma_s$, and $\sigma_{\hat s}^2=\sigma_s^2-\sigma_p^2$.
\end{proposition}

We choose $\phi^\star\in\arg\min_{\phi\in[\Gamma_0,\Gamma_1]} \lambda_{\mathrm{FP}}\mathrm{FPR}(\phi)+\lambda_{\mathrm{FN}}\mathrm{FNR}(\phi)$ with $\lambda_{\mathrm{FP}},\lambda_{\mathrm{FN}}\ge0$.

\begin{algorithm}[t]
\caption{Predictive Transition Detection}
\label{alg:predictive}
\begin{algorithmic}[1]\scriptsize
\Require $(\hat{\mathbf{x}}^s_{k|k},\mathbf{P}^s_{k|k})$, previous decision $\pi^s_{k-1}$, horizon $\ell$, thresholds $z^-,z^+$
\Ensure $\hat{\mathcal{T}}\in\{k,\ldots,k+\ell\}\cup\{\varnothing\}$
\State $(\hat{\mathbf{x}},\mathbf{P})\gets(\hat{\mathbf{x}}^s_{k|k},\mathbf{P}^s_{k|k})$
\For{$i=0,\ldots,\ell$}
    \State Compute $z_{k+i}$ and tentative $\pi_{k+i}$ from $(\hat{\mathbf{x}},\mathbf{P})$
    \If{$z_{k+i}\notin(z^-,z^+)$ and $\pi_{k+i}\neq\pi^s_{k-1}$}
        \State \Return $k+i$
    \ElsIf{$z_{k+i}\notin(z^-,z^+)$}
        \State \Return $\varnothing$
    \EndIf
    \State $\hat{\mathbf{x}}\gets\mathbf{A}\hat{\mathbf{x}}$, $\mathbf{P}\gets\mathbf{A}\mathbf{P}\mathbf{A}^\top+\mathbf{Q}$
\EndFor
\State \Return $\varnothing$
\end{algorithmic}
\end{algorithm}
Every time slot, the sensor first checks whether the current posterior supports a reliable decision change; otherwise, it propagates its local estimate up to $\ell$ slots. If Algorithm~\ref{alg:predictive} returns $\hat{\mathcal{T}}=k$, a predictive update $\mathcal{P}$ is transmitted immediately. If $\hat{\mathcal{T}}>k$, $\mathcal{P}$ is transmitted and a pending flag prevents redundant predictive updates before the predicted transition. If $\hat{\mathcal{T}}=\varnothing$, no predictive update is generated. For a state-$s$ sojourn, $I_s=\mathcal{T}_s-k_s^\star$, where $k_s^\star$ is the first $\mathcal{P}$ issue time.

\begin{figure}[tp]
\centering
\definecolor{curvecol}{HTML}{20639B}
\definecolor{disruptcol}{HTML}{ED553B}
\definecolor{recovcol}{HTML}{3CAEA3}
\definecolor{threshcol}{HTML}{888888}
\definecolor{xcol}{HTML}{D35400}

\resizebox{0.8\linewidth}{!}{%
\begin{tikzpicture}[scale=0.7, every node/.style={scale=0.65}, >=Stealth]
    \fill[disruptcol, opacity=0.1] (6.2, 0) rectangle (9.2, 4.8);
    \draw[disruptcol, thick, dashed] (6.2, 0) -- (6.2, 4.8);
    \node[font=\normalsize\bfseries, disruptcol!90!black] at (7.7, 4.5) {Disruption};
    \fill[recovcol, opacity=0.1] (9.2, 0) rectangle (12.0, 4.8);
    \draw[recovcol, thick, dashed] (9.2, 0) -- (9.2, 4.8);
    \draw[recovcol, thick, dashed] (12.0, 0) -- (12.0, 4.8);
    \node[font=\normalsize\bfseries, recovcol!90!black] at (10.6, 4.5) {Recovery};
    \draw[thick, ->] (0, -0.3) -- (0, 5.2);
    \draw[thick, ->] (-0.3, 0) -- (15.0, 0);
    \node[left, font=\normalsize\bfseries, black] at (0, 4.2) {State 1};
    \node[left, font=\normalsize\bfseries, black] at (0, 0.5) {State 0};
    \draw[thick, dashed, threshcol] (0, 2.2) -- (14.8, 2.2);
    \node[left, font=\normalsize\bfseries, black] at (0, 2.2) {$\Delta$};
    \draw[curvecol, very thick, line width=1.8pt, smooth, tension=0.5]
        plot coordinates {
            (0.5, 3.2) (1.5, 3.5) (2.5, 3.0) (3.5, 2.8)
            (4.5, 2.6) (5.3, 2.3) (6.0, 2.2)
            (6.8, 1.5) (7.8, 0.7) (8.8, 0.5)
            (9.8, 0.7) (11.0, 1.1) (12.0, 1.6)
            (12.8, 1.95) (13.6, 2.2) (14.5, 2.5)
        };
    \foreach \px/\py in {0.5/3.2, 1.5/3.5, 2.5/3.0, 3.5/2.8, 4.5/2.6, 12.3/1.75} {
        \draw[very thick, xcol, line width=1.4pt] (\px-0.12, \py-0.12) -- (\px+0.12, \py+0.12);
        \draw[very thick, xcol, line width=1.4pt] (\px-0.12, \py+0.12) -- (\px+0.12, \py-0.12);
    }
    \fill[white] (5.3, 2.3) circle (5pt);
    \draw[curvecol, very thick, line width=1.5pt] (5.3, 2.3) circle (5pt);
    \fill[curvecol] (6.0, 2.2) circle (3.5pt);
    \fill[white] (12.8, 1.95) circle (5pt);
    \draw[curvecol, very thick, line width=1.5pt] (12.8, 1.95) circle (5pt);
    \fill[curvecol] (13.6, 2.2) circle (3.5pt);
    \draw[thick, curvecol, line width=1pt] (5.3, 5.3) -- node[above, font=\large\bfseries, text=black] {$I$} (6.2, 5.3);
    \draw[very thick, curvecol, line width=1.2pt] (5.3, 5.15) -- (5.3, 5.45);
    \draw[very thick, curvecol, line width=1.2pt] (6.2, 5.15) -- (6.2, 5.45);
    \draw[dotted, thick, curvecol!50] (5.3, 2.3) -- (5.3, 5.1);
    \draw[dotted, thick, curvecol!50] (6.2, 4.8) -- (6.2, 5.1);
    \draw[->, thick, curvecol, line width=1pt] (8.5, 5.1) to[bend left=25] node[above, font=\normalsize\bfseries, text=black] {$\mathrm{AoI} > \theta$} (9.9, 5.1);
\end{tikzpicture}%
}



\caption{Predictive-plus-resilience transmission scheme. Predictive packets $\mathcal{P}$ ($\circ$) are issued at horizon $I$ prior to certified transitions ($\bullet$). Resilience packets $\mathcal{R}$ ($\times$) bound AoI, triggering outage recovery when $\mathrm{AoI}>\theta$.}
\label{fig:disruption-scheme}
\end{figure}
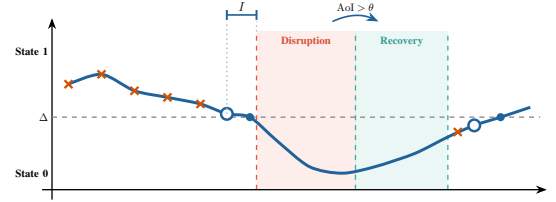

\section{Resilience-Driven Update Scheduling}
\label{sec:resilience}


Predictive triggering is fragile under outages: a correct $\mathcal{P}$ packet may still be blocked. We use $p_{\rm blk}$ as a renewal-reward approximation of the long-run fraction of time in which reception is unavailable. In one $0\to1\to0$ cycle, state $s$ contributes, in expectation, $p_r\bar B_s$ blocked slots, where $\bar B_s=\bar D_{\theta,s}+\bar t_h$. Hence
$p_{\rm blk}=
\min\left\{1,\,
\frac{p_r(\bar B_0+\bar B_1)}
{\mathbb{E}[T_0]+\mathbb{E}[T_1]}
\right\}.$
This treats a useful transition packet as sampling the long-run blackout process and assumes sparse, non-overlapping disruptions. Conditional on a decision-feasible $\mathcal{P}$ packet transmission, its effective communication-failure probability is approximated by
$\epsilon_{\rm eff}(p_t)=p_{\rm blk}+(1-p_{\rm blk})\bar\epsilon(p_t).$
Thus, the constraint $\epsilon_{\rm eff}(p_t)\le\epsilon_\ell$ is a tractable design choice for the communication-induced part of the lead-time objective, not an exact reformulation of \eqref{ct:time-trans}. The end-to-end lead-time probabilities are evaluated directly in Sec.~\ref{sec:sim}.

While in state $s$ and with no pending $\mathcal{P}$, the sensor attempts a resilience update $\mathcal{R}$ with probability $p_{u,s}$ in each eligible slot. Let
$\eta_s=p_{u,s}(1-\bar\epsilon(p_t))$
be the probability that an eligible slot produces a successful resilience update. We use a fixed AoI-tail tolerance $\epsilon_r\in(0,1)$, which controls how often the estimator is allowed to exceed the outage-detection threshold:
$\Prb(A_s>\theta_s)=(1-\eta_s)^{\theta_s}\le\epsilon_r.$
For the minimum refresh rate, this bound is tight, giving
$p_{u,s}^\star
=
\frac{1-\epsilon_r^{1/\theta_s}}
{1-\bar\epsilon(p_t)},$
provided $p_{u,s}^\star\le1$. Under the same geometric AoI approximation, the average residual detection delay is
$\bar D_{\theta,s}\approx\sum_{a=0}^{\theta_s-1}(\theta_s-a)\eta_s(1-\eta_s)^a$. Let $T_s$ be the state-$s$ sojourn and $I_s$ its prediction horizon. Resilience updates are attempted only in slots $1,\ldots,T_s-I_s-1$, because the first slot carries the new state indicator and slot $T_s-I_s$ is reserved for $\mathcal{P}$. Thus
$N_s^{\rm ref}|(T_s,I_s)\sim\mathrm{Binomial}((T_s-I_s-1)^+,p_{u,s})$.
With geometric tail $\Prb(T_s\ge n)=q_{ss}^{n-1}$ and independence of $T_s$ and $I_s$,
$\E[(T_s-I_s-1)^+]
=
\E[q_{ss}^{I_s}]
\frac{q_{ss}}{1-q_{ss}}.$
Hence, the average transmission rate  is
$r=
\frac{
2
+p_{u,0}\E[q_{00}^{I_0}]q_{00}/q_{01}
+p_{u,1}\E[q_{11}^{I_1}]q_{11}/q_{10}
}{
q_{01}^{-1}+q_{10}^{-1}
}.$
The moments $\E[q_{ss}^{I_s}]$ are estimated empirically via simulation.

The preceding analysis turns the proposed scheduling rule into explicit
design quantities: decision feasibility defines when a predictive packet
$\mathcal{P}$ is useful, the two-state surrogate gives the sojourn
statistics, and the AoI threshold controls outage resilience. We now
choose $p_t$, $\theta_0$, and $\theta_1$ to minimize long-term
transmission energy $np_t\,r$ while meeting the lead-time reliability constraint.

\begin{table}[t]
\centering
\caption{Validation of the Two-state Markov surrogate}
\vspace{-5pt}
\label{tab:markov_validation}
\scriptsize
\resizebox{0.48\textwidth}{!}{
\begin{tabular}{lcc|lcc}
\hline
\textbf{Metric} & \textbf{Surrogate} & \textbf{Empirical} & \textbf{Metric} & \textbf{Surrogate} & \textbf{Empirical} \\
\hline
$q_{01}$ & $0.06228$ & $0.062316$ & $q_{10}$ & $0.153705$ & $0.153951$ \\
$\mathbb{E}[T_0]$ & $16.0563$ & $16.0468$ & $\mathbb{E}[T_1]$ & $6.5060$ & $6.4956$ \\
\hline
\end{tabular}
}
\end{table}

\begin{problem}\label{prob:prob-1}
\begin{align}
\min_{p_t,\theta_0,\theta_1}\quad
& n p_t\,r(p_{u,0}^\star,p_{u,1}^\star)
\label{eq:opt_obj}\\
\mathrm{s.t.}\quad
& p_{\rm blk}+(1-p_{\rm blk})\bar\epsilon(p_t)\le\epsilon_\ell,
\label{ct:decision}\\
& p_{u,s}^\star\le1,\qquad s\in\{0,1\},
\label{ct:pu}\\
& \theta_s\in\mathbb{Z}_+,\qquad 0<p_t\le\rho.
\label{ct:domain}
\end{align}
\end{problem}

Constraint \eqref{ct:decision} bounds the delivery risk of a $\mathcal{P}$ packet, which may be blocked by an outage (probability $p_{\rm blk}$) or lost over the wireless link (probability $\bar\epsilon(p_t)$), and reduces to $\bar\epsilon(p_t)\le(\epsilon_\ell-p_{\rm blk})/(1-p_{\rm blk})$. Constraint \eqref{ct:pu} keeps the refresh probability $p_{u,s}^\star$ valid, equivalent to $\bar\epsilon(p_t)\le\epsilon_r^{1/\theta_s}$ for each $s$.

We solve the problem by a grid search over $(\theta_0,\theta_1)$ coupled with a one-dimensional search over $p_t$. Each grid point yields an admissible packet-error level $\bar\epsilon_{\max}=\min\{(\epsilon_\ell-p_{\rm blk})/(1-p_{\rm blk}),\epsilon_r^{1/\theta_0},\epsilon_r^{1/\theta_1}\}$, and is discarded if $p_{\rm blk}\ge\epsilon_\ell$ or the power budget cannot meet $\bar\epsilon(p_t)\le\bar\epsilon_{\max}$. Otherwise, since $\bar\epsilon(p_t)$ decreases monotonically in $p_t$, the lowest-energy feasible choice is the smallest $p_t$ satisfying $\bar\epsilon(p_t)=\bar\epsilon_{\max}$. The procedure exposes the energy trade-off: smaller $\theta_s$ detects outages faster and reduces $p_{\rm blk}$ but inflates the refresh rate, while larger $p_t$ lowers both packet errors and the required refresh probability at the cost of more expensive transmissions.

Algorithm~\ref{alg:sensor-full} combines the two update types. The pending flag prevents repeated predictive packets for the same anticipated crossing. Resilience updates are used only when no reliable transition is currently predicted; hence, their role is not to improve estimation uniformly, but to keep the remote agent fresh enough for outage detection and recovery.

\begin{algorithm}[t]
\caption{Sensor Transmission Strategy}
\label{alg:sensor-full}
\begin{algorithmic}[1]\scriptsize
\Require $p_{u,0}^\star,p_{u,1}^\star,\ell,z^-,z^+,\phi$
\State Initialize pending flag $F\gets0$ and previous local decision $\pi^s_{-1}$
\For{each slot $k=0,1,\ldots$}
    \State Run the local Kalman update and compute $\pi^s_k$ using $z^\pm$ and $\phi$
    \If{$\pi^s_k\ne\pi^s_{k-1}$} \State $F\gets0$ \Comment{new sojourn} \EndIf
    \If{$F=1$} \State $\pi^s_{k-1}\gets\pi^s_k$; \textbf{continue} \EndIf
    \State $\hat{\mathcal{T}}\gets$ Algorithm~\ref{alg:predictive}$(\hat{\mathbf{x}}^s_{k|k},\mathbf{P}^s_{k|k},\pi^s_{k-1},\ell)$
    \If{$\hat{\mathcal{T}}\ne\varnothing$}
        \State Transmit $\mathcal{P}$; \textbf{if} $\hat{\mathcal{T}}>k$ \textbf{then} $F\gets1$
    \Else
        \State Transmit $\mathcal{R}$ with probability $p_{u,\pi^s_k}^\star$
    \EndIf
    \State $\pi^s_{k-1}\gets\pi^s_k$
\EndFor
\end{algorithmic}
\end{algorithm}

\section{Simulation and Results}
\label{sec:sim}
\subsection{Validation of the Markov Surrogate}
\label{subsec:markov_validation}

Before evaluating the communication policies, we validate the two-state surrogate used in Approximation \ref{thm:switching}. From the resulting binary sequence, we estimate the empirical switching probabilities $\hat q_{01}$ and $\hat q_{10}$, and the empirical sojourn lengths $\hat T_0$ and $\hat T_1$. Fig.~\ref{fig:markov_validation} compares the empirical sojourn survival functions with the geometric tails $\Pr(T_s\ge t)=q_{ss}^{t-1}$ implied by the surrogate. The thresholded Gaussian process is not exactly a first-order Markov chain, so some tail deviations are expected because the continuous state retains memory beyond the binary label. Nevertheless, Table~\ref{tab:markov_validation} shows that the surrogate accurately captures the switching probabilities and mean state durations used in the long-term rate and blackout calculations. Thus, the surrogate is not used to claim exact binary Markovianity, but to obtain closed-form renewal quantities whose dominant terms are matched to the true process.

Thus, Problem~\ref{prob:prob-1} is surrogate-optimal for the approximation-driven model, and residual metric gaps reflect the Gaussian-posterior, Markov-sojourn, and blackout-fraction surrogates.

\subsection{Performance Evaluation}



\begin{table*}[t]
\centering
\caption{Quality Metrics Comparison at Matched Energy}
\vspace{-5pt}
\label{tab:benchmark}
\scriptsize
\renewcommand{\arraystretch}{0.90}
\resizebox{0.6\textwidth}{!}{
\begin{tabular}{llcccccc}
\hline
& \textbf{Metric} & \textbf{Proposed} & \textbf{Ideal} & \textbf{Predictive Only} & \textbf{Event-based} & \textbf{AoI-based} & \textbf{AoII-based} \\
\hline
\multirow{4}{*}{\textbf{Par.~A}}
& FPR & $0.0181$ & $0.0124$ & $0.0475$ & $0.0410$ & $0.0171$ & $0.0259$ \\
& FNR & $0.0613$ & $0.0077$ & $0.1316$ & $0.1305$ & $0.0324$ & $0.0895$ \\
& $\widehat{\Prb}(L\geq0)$ & $0.9103$ & $1.0000$ & $0.8632$ & $0.8690$ & $0.8874$ & $0.9115$ \\
& $\widehat{\Prb}(L>0)$ & $0.7793$ & $1.0000$ & $0.7517$ & $0.0000$ & $0.0000$ & $0.0000$ \\
\hline
\multirow{6}{*}{\textbf{Par.~B}}
& FPR & $0.0170$ & $0.0124$ & $0.0151$ & $0.0145$ & $0.0122$ & $0.0143$ \\
& FNR & $0.0574$ & $0.0077$ & $0.1665$ & $0.1957$ & $0.0175$ & $0.1537$ \\
& $\widehat{\Prb}(L\geq0)$ & $0.9138$ & $1.0000$ & $0.8632$ & $0.8690$ & $0.9552$ & $0.8989$ \\
& $\widehat{\Prb}(L>0)$ & $0.8011$ & $1.0000$ & $0.7598$ & $0.0161$ & $0.7989$ & $0.0425$ \\
\hline
\end{tabular}
} 
\vspace{-0.1cm}
\end{table*}

\begin{figure}[t]
    \centering
    \includegraphics[width=0.7\linewidth]{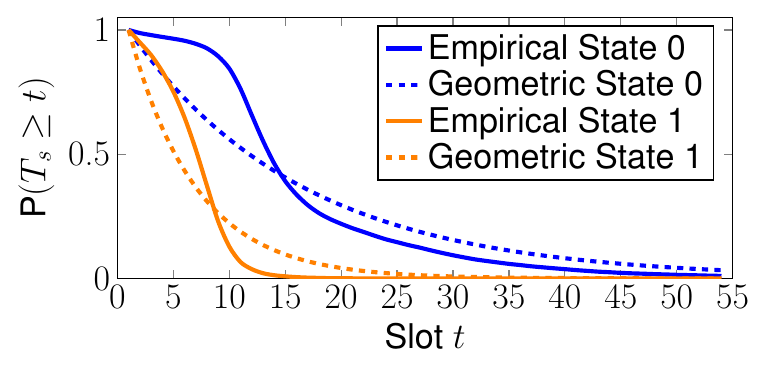}
    \vspace{-8pt}
    \caption{Validation of the Two-State Markov surrogate}
    \vspace{-2pt}
    \label{fig:markov_validation}
\end{figure}

\begin{table}[t]
\centering
\caption{Simulation and optimized parameters}
\vspace{-5pt}
\scriptsize
\setlength{\tabcolsep}{3pt}
\renewcommand{\arraystretch}{0.95}
\resizebox{0.95\columnwidth}{!}{%
\begin{tabular}{@{}lp{7.5cm}@{}}
\hline
\multicolumn{2}{c}{\textbf{Simulation parameters}}\\
\hline
System matrices & $\mathbf{A}=[0\;1;\,-0.9\;1.8]$, $\mathbf{C}=[0.5\;1.0]$, $\mathbf{Q}=\mathrm{diag}(0,1)$, $\mathbf{R}=0.1$ \\
Threshold & $\Delta=4.0$, $\epsilon_\ell=0.1$, $\ell=10$, $\alpha_{\mathrm{FP}}=\alpha_{\mathrm{FN}}=0.05$ \\
Channel & $n=128$, $l=256$, $\sigma^2=0.4$, $p_t\in[0.05,200]$ mW \\
Outage & $p_r=0.05$, $\lambda=3.0$, $\kappa=2.0$ \\
Simulation time & $5\times 10^4$ \\
\hline
\multicolumn{2}{c}{\textbf{Inferred optimum}}\\
\hline
Markov surrogate & $q_{01}=0.062$, $q_{10}=0.154$ \\
Aux.\ threshold/budgets & $\phi=3.859$, $\boldsymbol\theta=[13,3]$ \\
Refresh probabilities & $p_{u,0}=0.3184$, $p_{u,1}=0.8375$ \\
Prediction horizons & $\E[I_0]=1.08$, $\E[I_1]=1.30$ \\
Optimal power / avg.\ PER & $p_t^\star=40.37$, $\bar\epsilon=0.0632$ \\
\hline
\end{tabular}}
\label{tab:key_para}
\end{table}

\begin{figure}[t]
    \vspace{-3pt}
    \hspace{-15pt}
    \centering
    \includegraphics[width=1.05\linewidth]{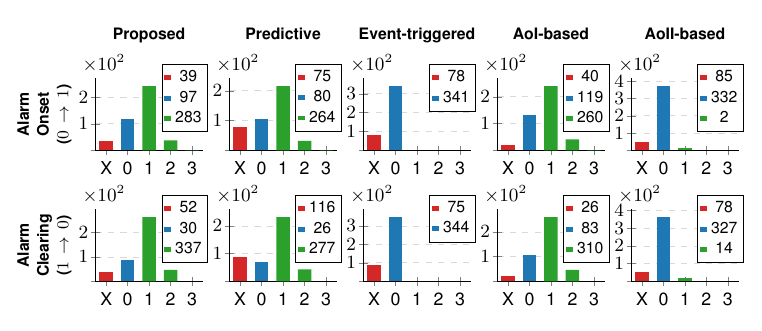}
    \vspace{-0.8cm}
    \caption{Comparison in update lead time ($L$) among policies.}
    \label{fig:histogram}
    \vspace{-0.01cm}
\end{figure}

We compare the proposed policy with four baselines under matched average transmit energy, $\mathcal{E}=np_t^\star r(p_{u,0}^\star,p_{u,1}^\star)$. The main performance metric is not only whether the decision is on time, measured by $\widehat{\Prb}(L\ge0)$, but whether it arrives with positive reaction time, measured by $\widehat{\Prb}(L>0)$. Table~\ref{tab:key_para} lists the simulation parameters and the inferred optimum of Problem~\ref{prob:prob-1}. The baselines are: \emph{Ideal}, a noise- and outage-free lower bound; \emph{Event-triggered}, transmission only on detected state changes; \emph{Predictive-only}, Algorithm~\ref{alg:predictive} without $\mathcal{R}$; \emph{AoI-based}, state-agnostic Bernoulli scheduling minimizing average AoI; and \emph{AoII-based}, transmission during the first $W_s$ slots of each state.

Table~\ref{tab:benchmark} reports the two remote-agent paradigms of Sec.~\ref{sec:agent}. Under \emph{decision adoption} (Par.~A), the agent simply copies the last received label, so the three baselines that never transmit an anticipated label: Event-triggered, AoI-based, and AoII-based, yield $\widehat{\Prb}(L>0)=0$ by construction. On decision accuracy, the proposed scheme trails AoI-based only marginally on FPR/FNR (which benefits from continually fresh estimates) and clearly outperforms the remaining baselines. On lead-time, it attains the highest non-ideal $\widehat{\Prb}(L>0)$ and is essentially tied with AoII-based on $\widehat{\Prb}(L\geq0)$, making it the only non-ideal scheme that simultaneously delivers a high probability of timely arrival and substantial early-decision mass.

Under \emph{filter-based decision} (Par.~B), every scheme can run the same posterior predictor at the estimator, so the comparison isolates how scheduling shapes the posterior at decision instants. AoI-based now leads on FPR/FNR and $\widehat{\Prb}(L\geq0)$: uniform freshness keeps the estimator both accurate everywhere and quick to detect disruptions. The proposed policy trails it on these accuracy metrics but attains the highest $\widehat{\Prb}(L>0)$ across all schemes, because concentrating transmissions around certified transitions converts what would otherwise be on-time updates into early ones. Fig.~\ref{fig:histogram} makes this trade-off visible: for both alarm onset ($0\to1$) and clearing ($1\to0$), the proposed scheme accepts slightly more missed transitions than AoI-based but redistributes the remaining mass toward $L>0$, where decisions are early enough to be actionable. Predictive-only also produces early updates but misses more often because long sojourns with a blocked $\mathcal{P}$ packet go undetected, while Event-triggered and AoII-based concentrate mass at $L=0$ since neither reserves transmissions for an anticipated crossing.

\begin{figure}[tb]
\vspace{-6pt}
\centering
\subfloat[$\Pr(L \geq 0)$, $\varepsilon_\ell = 0.1$]{
    \hspace{-15pt}
    \includegraphics[width=0.47\linewidth]{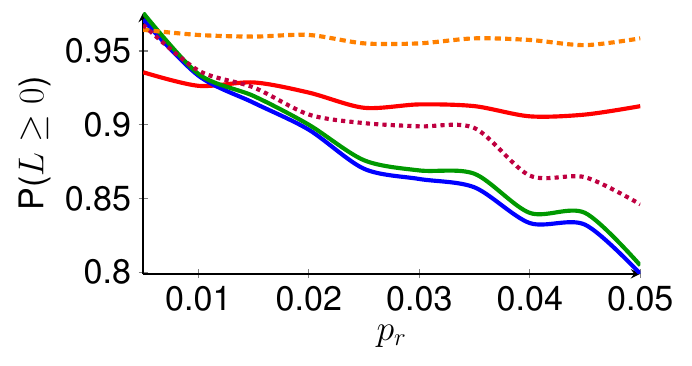}
    \label{fig:insight-a}
    \vspace{-10pt}
}
\subfloat[$\Pr(L > 0)$, $\varepsilon_\ell = 0.1$]{
    \hspace{-15pt}
    \includegraphics[width=0.47\linewidth]{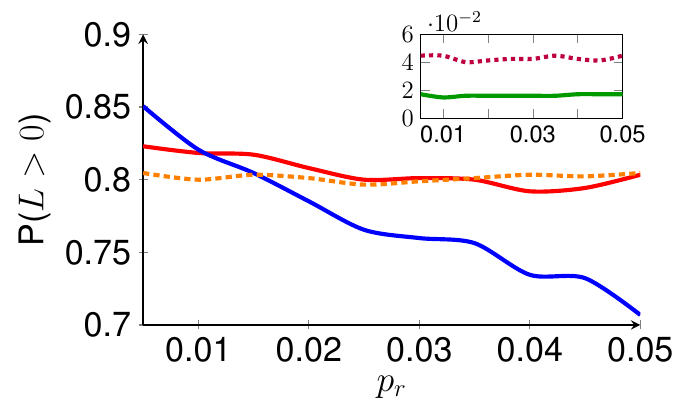}
    \label{fig:insight-b}
    \vspace{-10pt}
}\\[-10pt]
\subfloat[$\Pr(L \geq 0)$, $p_r = 0.01$]{
    \hspace{-15pt}
    \includegraphics[width=0.47\linewidth]{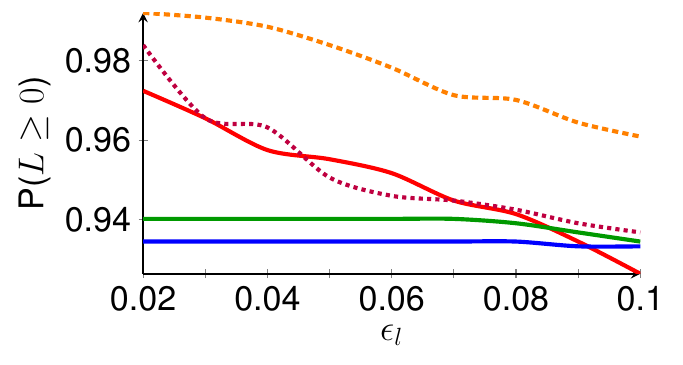}
    \label{fig:insight-c}
    \vspace{-15pt}
}
\subfloat[$\Pr(L > 0)$, $p_r = 0.01$]{
    \hspace{-15pt}
    \includegraphics[width=0.47\linewidth]{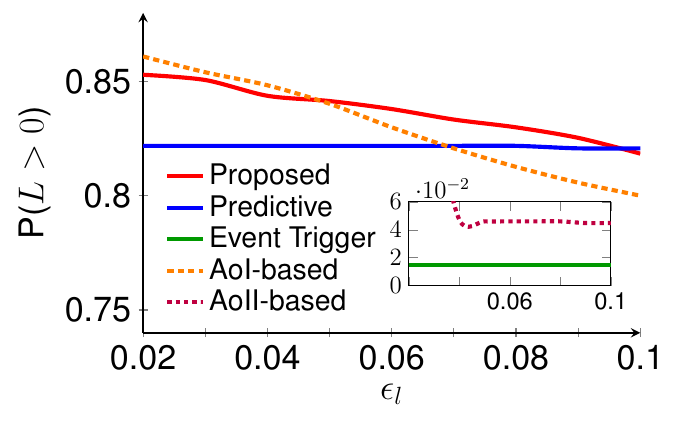}
    \label{fig:insight-d}
    \vspace{-12pt}
}
\caption{Lead-time success probability under varying $p_r$ and $\varepsilon_\ell$.}
\label{fig:insights}
\end{figure}

Fig.~\ref{fig:insights} sweeps the outage probability $p_r$ at fixed $\epsilon_\ell=0.1$ (panels a,b) and the reliability budget $\epsilon_\ell$ at fixed $p_r=0.01$ (panels c,d), and shows that no single baseline wins everywhere. AoI-based dominates $\widehat{\Prb}(L\geq0)$ throughout the swept range thanks to its uniform freshness. At low disruption ($p_r\lesssim 0.01$), where outage protection brings little benefit, Predictive-only and Event-triggered also reach the top on $\widehat{\Prb}(L\geq0)$, and Predictive-only additionally leads $\widehat{\Prb}(L>0)$. At the two extremes: strict $0.01\lesssim \epsilon_\ell\lesssim 0.04$ or high $p_r\gtrsim 0.05$, AoI-based takes over $\widehat{\Prb}(L>0)$ as well, since the energy budget then forces it close to full-slot transmission. The proposed policy is the only one that remains competitive across all regimes, and is the strongest choice for $\widehat{\Prb}(L>0)$ in the practically common middle band of moderate disruption ($p_r\in[0.01,0.04]$) combined with a relaxed budget ($\epsilon_\ell\in[0.04,0.1]$).

\section{Conclusion}

This letter proposed a decision-aware transmission strategy for remote threshold decisions over unreliable short-packet links. Posterior FPR/FNR constraints first define when a binary action is decision-feasible. This feasibility then drives predictive transition packets, while AoI-controlled resilience packets protect them from outage-induced staleness. Simulations show that the main benefit is not uniformly better estimation, but a higher probability of actionable lead time, $\widehat{\Prb}(L>0)$, at matched energy. Thus, threshold-monitoring communication should be designed around reliable decisions before the crossing, not only around accuracy or freshness.

\vspace{-6pt}
\bibliographystyle{IEEEtran}
\bibliography{references}

\end{document}